\documentclass[reqno,12pt]{amsart}

\usepackage[T1]{fontenc}
\usepackage{amsmath,amssymb,amsthm,mathtools}
\usepackage[a4paper,margin=1.02in]{geometry}
\usepackage{microtype}
\usepackage[colorlinks=true,linkcolor=blue,citecolor=blue,urlcolor=blue]{hyperref}

\newtheorem{theorem}{Theorem}[section]

\newtheorem{lemma}[theorem]{Lemma}
\newtheorem{corollary}[theorem]{Corollary}
\theoremstyle{remark}

\newcommand{\T}{\mathbb T}
\newcommand{\Z}{\mathbb Z}
\newcommand{\R}{\mathbb R}

\newcommand{\Q}{\mathbb Q}

\newcommand{\dist}{\operatorname{dist}}

\newcommand{\normT}[1]{\left\lVert #1\right\rVert_{\T}}
\newcommand{\vct}[2]{\begin{pmatrix}#1\\#2\end{pmatrix}}
\newcommand{\eone}{\begin{pmatrix}1\\0\end{pmatrix}}
\newcommand{\etwo}{\begin{pmatrix}0\\1\end{pmatrix}}

\title[A doubled Gordon Argument]{A doubled Gordon threshold for palindromic quasiperiodic Schr\"odinger operators}
\author{Wencai Liu}
\address{
	Department of Mathematics, Texas A\&M University, College Station, TX, 77843, USA.
}
\email{liuwencai1226@gmail.com, wencail@tamu.edu}
\date{ \today}

\begin{document}

\begin{abstract}
We consider one-frequency quasiperiodic Schr\"odinger operators
\[
(H_{v,\alpha,\theta}u)(n)
=
u(n+1)+u(n-1)
+
v(\theta+n\alpha)u(n)
\]
acting on $\ell^2(\mathbb Z)$, where $\alpha\notin\mathbb Q$ and
$v\in C^2(\mathbb T,\mathbb R)$ is an even function. We develop
a new Gordon-type method that exploits approximate repetitions and
palindromic symmetries simultaneously. Denote by $L(E)$ the Lyapunov exponent and let 

\[\beta(\alpha)
=
\limsup_{|k|\to\infty}
-\frac{\log\|k\alpha\|_{\mathbb R/\mathbb Z}}{|k|}.
\]
We prove that, for every
completely resonant phase
$2\theta\in\alpha\mathbb Z+\mathbb Z$,
  $E$ cannot be an eigenvalue if
$L(E)<2\beta(\alpha)$.

As an application, consider the almost Mathieu operator
\[
(H_{\lambda,\alpha,\theta}u)(n)
=
u(n+1)+u(n-1)
+
2\lambda\cos\bigl(2\pi(\theta+n\alpha)\bigr)u(n).
\]
We show that if
$2\theta\in\alpha\mathbb Z+\mathbb Z$ and
$1<|\lambda|<e^{2\beta(\alpha)}$, then
$H_{\lambda,\alpha,\theta}$ has purely singular continuous spectrum.
This resolves the remaining absence-of-eigenvalues part of a
conjecture of Avila and Jitomirskaya concerning the sharp spectral
transition for completely resonant phases.
\end{abstract}
\maketitle

\section{Introduction and main results}

 In this paper,   we study the one-frequency quasiperiodic
Schr\"odinger operator
\begin{equation}\label{defop}
    (H_{v,\alpha,\theta}u)(n)
=
u(n+1)+u(n-1)
+
v(\theta+n\alpha)u(n)
\end{equation}
acting on $\ell^2(\Z)$, where $\alpha\in\R\setminus\Q$ is the
frequency, $\theta\in\T=\R/\Z$ is the phase, and
$v\colon\T\to\R$ is a   sampling function. Throughout this
paper, we assume that $v\in C^2(\T,\R)$ and   $v$ is even:
\begin{equation*}
    v(-x)=v(x),
    \qquad x\in\R.
\end{equation*}

For $x\in\R$, let
\[
    \normT{x}=\dist(x,\Z).
\]
The exponential strength of the frequency resonances is measured by
\begin{equation*}
    \beta(\alpha)
    =
    \limsup_{|k|\to\infty}
    -\frac{\log\normT{k\alpha}}{|k|}
    \in[0,\infty].
\end{equation*}

For $E\in\R$, define the one-step transfer matrix
\begin{equation*}
    A^E(x)
    =
    \begin{pmatrix}
        E-v(x)&-1\\
        1&0
    \end{pmatrix},
    \qquad x\in\T,
\end{equation*}
and, for $n\geq1$, define
\begin{equation}\label{eq:n-step}
    A^E_n(x)
    =
    A^E(x+(n-1)\alpha)\cdots
    A^E(x+\alpha)A^E(x),
    \qquad
    A^E_0(x)=I.
\end{equation}
The Lyapunov exponent associated with the Schr\"odinger cocycle
$(\alpha,A^E)$ is
\begin{equation*}
    L(E)
    =
    \lim_{n\to\infty}
    \frac1n\int_{\T}\log\|A^E_n(x)\|\,dx.
\end{equation*}

We call $\theta$  {completely resonant} with respect to $\alpha$
if
$ 2\theta\in\alpha\Z+\Z.$

Completely resonant phases play a crucial role in the study of  spectral theory
of quasiperiodic Schr\"odinger operators. On the Aubry-dual side, they
are related to rational fibered rotation numbers and arise naturally
in the study of spectral   edges and gap opening; see, for example,
\cite{aj09,puig}. We refer readers to a recent survey for further
discussion. \cite{jicm}.

Our main result is the following:

\begin{theorem}\label{thm:main}
Assume that $v\in C^2(\T,\R)$ is an  even function and
$\alpha\in\R\setminus\Q$. Assume that $\theta$ is completely
resonant, that is, $ 2\theta\in\alpha\Z+\Z.$

If $L(E)<2\beta(\alpha),$ 
then the equation $ H_{v,\alpha,\theta}u=Eu$
has no nonzero $\ell^2(\Z)$ solution $u$. 
\end{theorem}

\subsection{Repetitions and palindromic symmetries}

There are two   types of   symmetry that can lead to the absence of eigenvalues for 
  one-dimensional Schr\"odinger operators. The first is approximate repetition symmetry:
\begin{equation}\label{eq:repetition-intro}
    V(n+q)\approx V(n),
\end{equation}
and the second is approximate reflection symmetry:
\begin{equation}\label{eq:reflection-intro}
    V(q-n)\approx V(n).
\end{equation}
Potentials of the form
\eqref{eq:repetition-intro} are commonly called \emph{Gordon-type
potentials}, while potentials of the form \eqref{eq:reflection-intro} are
often called \emph{palindromic type potentials}.

Gordon-type and palindromic arguments have a long history in the
study of the absence of eigenvalues; see
\cite{gor,sim1,as82,js,jl,jl1,jl2,ayz,liupafa}
and the references therein. Sharp quantitative 
arguments compare the exponential accuracy of the corresponding
  symmetry with the exponential growth rate of transfer
matrices, measured by the Lyapunov exponent.

For a quasiperiodic potential $    V(n)=v(\theta+n\alpha),
$
a frequency resonance $\normT{q\alpha}\ll1$ produces an approximate
repetition of the potential:
\[
    V(n+q)-V(n)
    =
    O\bigl(\normT{q\alpha}\bigr).
\]
 As a result, a sharp Gordon argument applied to
this repetition alone leads naturally to the condition
\[
    L(E)<\beta(\alpha).
\]
Similarly, palindromic arguments applied separately to reflective
repetitions see only one exponentially small resonance factor and
lead to the same single-resonance threshold in the completely
resonant setting. We refer to \cite{liupafa} for a recent discussion
of these two mechanisms.

It has remained open to improve the threshold
$L(E)<\beta(\alpha)$
to the doubled threshold $L(E)<2\beta(\alpha)$.

In the present paper we develop a new approach to treat repetitions and palindromic symmetries 
simultaneously.  Let us take $\theta=0$  as an example to explain the new idea in our proof.
In this case,  the evenness of $v$ implies that the potential $V(n)=v(n\alpha)$  is reflection symmetric about the origin:
\[
    V(-n)=V(n).
\]

Note that the frequency resonances produce  repetition of the potential. Thus, at $\theta=0$, the potential has both a repetition structure and a
palindromic structure.

Our approach exploits these two structures simultaneously. Let
\[
    h=q\alpha-p\in[-1/2,1/2],
    \qquad p\in\Z,
    \qquad |h|=\normT{q\alpha}.
\]
A one-sided comparison of neighboring transfer matrices involves the
first difference
\[
    ||A^E_q(x+h)-A^E_q(x)||\text{ and/or  } ||A^E_q(x-h)-A^E_q(x)||
\]
whose   upper bounds contain only one factor of $|h|$. This is
the mechanism behind the usual threshold $L(E)<\beta(\alpha)$.

In contrast to the usual one-sided Gordon comparison, our main new
idea is to compare the two neighboring blocks symmetrically:
\begin{equation}\label{eq:symmetric-second-difference-intro}
   || A^E_q(x+h)+A^E_q(x-h)-2A^E_q(x)||.
\end{equation}

The first-order errors cancel in
\eqref{eq:symmetric-second-difference-intro}, leaving the quadratic
factor
\[
    |h|^2=\normT{q\alpha}^2.
\]

This quadratic gain is the source of the doubled arithmetic threshold
$L(E)<2\beta(\alpha)$.
 
There is a second important difference from the standard Gordon and
palindromic arguments. In the traditional approaches, the transfer matrix block
relations are typically applied to the  eigenfunction at site $0$, namely $U(0)$. In our
argument, however, the symmetric comparison is applied after one
transfer block, namely to $U(q)=A_q^E(x)U(0).$
 
Indeed, the matrices $A_q^E(x-h)$, $A_q^E(x)$, and
$A_q^E(x+h)$ represent three consecutive  transfer blocks, corresponding
respectively to the sites from $-q$ to $0$, from $0$ to $q$, and
from $q$ to $2q$. Thus, $U(q)$  locates in the middle of two consecutive blocks. Applying the symmetric comparison at
this intermediate site allows the approximate repetition and the
palindromic symmetry to interact with the transfer dynamics in a way
that is not visible when one works only with $U(0)$. This  new idea is another key ingredient in our proof.

\subsection{Application to the almost Mathieu operator}

Letting $v(x)=2\lambda\cos(2\pi x)$
in~\eqref{defop}, we obtain the almost Mathieu operator
\[
   (H_{\lambda,\alpha,\theta}u)(n)
=
u(n+1)+u(n-1)
+
2\lambda\cos\bigl(2\pi(\theta+n\alpha)\bigr)u(n).
\]
For every energy in the spectrum of the almost Mathieu operator, its
Lyapunov exponent is $  L(E)=\max\{\log|\lambda|,0\}.$

Immediately,  Theorem~\ref{thm:main}  gives the following
corollary:

\begin{corollary}\label{cor:amo}
Let $\alpha\in\R\setminus\Q$ and assume that $  2\theta\in\alpha\Z+\Z.$

If $1<|\lambda|<e^{2\beta(\alpha)},$
then $\sigma_{\mathrm{pp}}
    \bigl(H_{\lambda,\alpha,\theta}\bigr)
    =
    \varnothing.$
 
\end{corollary}

The exponent $2\beta(\alpha)$ is the expected sharp transition
threshold at completely resonant phases. More precisely, Avila and Jitomirskaya \cite{aj09} conjectured that for completely resonant phase $2\theta\in\alpha\Z+\Z$,

  \begin{description}
      \item[1] $H_{\lambda,\alpha,\theta}$ has purely singular continuous spectrum when $1<|\lambda|<e^{2\beta(\alpha)}$;
      \item[2] $H_{\lambda,\alpha,\theta}$ has  Anderson localization when $|\lambda|>e^{2\beta(\alpha)}$.
  \end{description}

The localization part was proved in
\cite{liupmj}: for every completely resonant phase,
$H_{\lambda,\alpha,\theta}$ satisfies Anderson localization whenever $|\lambda|>e^{2\beta(\alpha)}$.

The complementary absence-of-eigenvalues regime remained open.
Corollary~\ref{cor:amo} resolves this remaining part. Since the
almost Mathieu operator has no absolutely continuous spectrum in
the supercritical regime $|\lambda|>1$, the absence of eigenvalues
implies that its spectrum is purely singular continuous whenever $1<|\lambda|<e^{2\beta(\alpha)}$ and
$ 2\theta\in\alpha\Z+\Z$.
Therefore, Corollary~\ref{cor:amo}, together with the main result in
\cite{liupmj}, establishes the sharp spectral transition
 for completely resonant phases.

\section{Basics}
In this section, we  collect some basic facts in this area.

Fix $E\in\R$.
For a solution of
\begin{equation}\label{geig}
    H_{v,\alpha,\theta}\varphi=E\varphi,
\end{equation}
let
\begin{equation*}
\Phi(n)=\vct{\varphi(n)}{\varphi(n-1)}.
\end{equation*}
The eigen-equation \eqref{geig} is equivalent to
\begin{equation*}
 \Phi(n+1)=A^E(\theta+n\alpha)\Phi(n).
\end{equation*}
Thus,
\begin{equation}\label{gtransfer}
 \Phi(m+n)=A^E_n(\theta+m\alpha)\Phi(m),
 \qquad m\in\Z,\quad n\geq0.
\end{equation}
Moreover,
\begin{equation*}
 \det A^E(x)=1,
 \qquad
 \det A^E_n (x)=1.
\end{equation*}
When there is no ambiguity, we drop the dependence of $E$, $v$, $\alpha$ and $\theta$.

\begin{lemma}
\label{lem:uniform-upper}
For every $\varepsilon>0$, there is a constant
$C_{\varepsilon,E,v}>0$ such that
\begin{equation*}
 \sup_{x\in\T}\|A^E_{n}(x)\|
 \leq C_{\varepsilon,E,v}
 e^{(L(E)+\varepsilon)n},
 \qquad n\geq0.
\end{equation*}
In particular, for all sufficiently large $n$,
\begin{equation*}
 \sup_{x\in\T}\|A_{n}^E(x)\|
 \leq e^{(L(E)+\varepsilon)n}
\end{equation*}

\end{lemma}

The gain from $\beta(\alpha)$ to $2\beta(\alpha)$ comes from a symmetric
second difference of the full transfer matrix.

\begin{lemma}
\label{lem:second-derivative}
Assume $v\in C^2(\T,\R)$.  For every $\varepsilon>0$, there is a constant
$C=C(E,v,\varepsilon)>0$ such that
\begin{equation}\label{eq:second-derivative-bound}
 \sup_{x\in\T}\|(A_{q}^E)''(x)\|
 \leq Cq^2e^{(L(E)+\varepsilon)q},
 \qquad q\geq1.
\end{equation}
Moreover, for   every   $h$ with $|h|\leq1/2$,
\begin{equation}\label{eq:symmetric-second-difference}
  \sup_{x\in\T} \bigl\|A_{q}^E(x+h)+A_{q}^E(x-h)-2A_{q}^E(x)\bigr\|
 \leq Cq^2h^2e^{(L(E)+\varepsilon)q}.
\end{equation}
\end{lemma}

\begin{proof}
Write $A(x)=A^E(x)$.
Direct computations show
\[
 A'(x)=
 \begin{pmatrix}-v'(x)&0\\0&0\end{pmatrix},
 \qquad
 A''(x)=
 \begin{pmatrix}-v''(x)&0\\0&0\end{pmatrix}.
\]
Differentiating the product \eqref{eq:n-step} once gives
\begin{equation*}
 A_q'(x)
 =\sum_{j=0}^{q-1}
 A_{q-j-1}(x+(j+1)\alpha)
 A'(x+j\alpha)A_j(x).
\end{equation*}
Differentiating a second time gives
\begin{align*}
 A_q''(x)
 ={}&\sum_{j=0}^{q-1}
 A_{q-j-1}(x+(j+1)\alpha)
 A''(x+j\alpha)A_j(x)
 \\
 &+2\sum_{0\leq i<j\leq q-1}
 A_{q-j-1}(x+(j+1)\alpha)
 A'(x+j\alpha)\notag\\
 &\hspace{25mm}\cdot
 A_{j-i-1}(x+(i+1)\alpha)
 A'(x+i\alpha)A_i(x).\notag
\end{align*}
  By Lemma~\ref{lem:uniform-upper}, we obtain
\eqref{eq:second-derivative-bound} and then
\eqref{eq:symmetric-second-difference}.
\end{proof}

 \begin{lemma}
\label{lem:simplicity}
For a discrete Schr\"odinger operator $H=\Delta+V$ on $\ell^2(\Z)$, every eigenvalue is simple.
\end{lemma}

\subsection{Reduction of a completely resonant phase}
Given $2\theta\in \alpha\Z+\Z$,
 by shifting the operator $H_{v,\alpha,\theta}$ to $H_{v,\alpha,\theta \pm \alpha}$,
it suffices to  study the four canonical phases
\begin{equation*}
 \theta=s
 \quad\text{or}\quad
 \theta=s+\frac{\alpha}{2},
 \qquad
 s\in\left\{0,\frac12\right\}.
\end{equation*}
In our proof, we   group the four canonical phases into two cases:
\begin{description}
    \item[\bf Case I:]$\theta=s$,
 $s\in\left\{0,\frac12\right\}$;
 \item[\bf Case II:]$\theta=s+\frac{\alpha}{2}$,
 $s\in\left\{0,\frac12\right\}$.
\end{description}

Evenness and periodicity imply that both $0$ and $1/2$ are reflection centers
of $v$: 
\begin{equation}\label{eq:symmetry-centers}
 v(s+x)=v(s-x), \text { for all } x,
 \qquad
 s\in\left\{0,\frac12\right\}.
\end{equation}
In particular,
\begin{equation}\label{eq:vprime-zero}
 v'(s)=0,
 \qquad s\in\left\{0,\frac12\right\}.
\end{equation}

Fix an energy $E$ satisfying
\begin{equation*}
 L=L(E)<2\beta(\alpha).
\end{equation*}

Choose  $\gamma$   so that 
\begin{equation}\label{eq:gamma-epsilon}
 \frac{L}{2}<\gamma<\beta(\alpha).
\end{equation}

By the definition of $\beta(\alpha)$, there are infinitely many positive
integers $q\to\infty$ and integers $p$ such that  
\begin{equation*}
 h=h_q=q\alpha-p\in[-1/2,1/2],
\end{equation*}
and
\begin{equation}\label{eq:h-small}
 |h|=\normT{q\alpha}\leq e^{-\gamma q}.
\end{equation}

\section{Proof of Case I}

Fix
\begin{equation*}
 \theta=s,
 \qquad s\in\left\{0,\frac12\right\}.
\end{equation*}
By \eqref{eq:symmetry-centers},
\[
 v(s-n\alpha)=v(s+n\alpha).
\]
Denote by $V(n)=v(s+n\alpha)$, so that $H_{v,\alpha,s}=\Delta+V$.
The potential $V=\{V(n)\}$ satisfies for all $n\in\Z$, 
\begin{equation}\label{gsym}
    V(n)=V(-n).
\end{equation}

Denote by
\begin{equation}\label{eq:site-blocks}
 T_0=A_q(s),\;\;
 T_+=A_q(s+h)=A_q(s+q\alpha),\;\;
 T_-=A_q(s-h)=A_q(s-q\alpha).
\end{equation}
By \eqref{gtransfer},
\begin{equation}\label{eq:site-block-action}
 \Phi(q)=T_0\Phi(0),
 \qquad
 \Phi(2q)=T_+\Phi(q),
 \qquad
 \Phi(0)=T_-\Phi(-q).
\end{equation}

Introduce
\begin{equation*}
 R(x)=
 \begin{pmatrix}
 1&0\\
 E-v(x)&-1
 \end{pmatrix}.
\end{equation*}
The eigenvalue equation \eqref{geig} gives
\begin{equation}\label{eq:R-action}
 R(\theta+n\alpha)\Phi(n)
 =\vct{\varphi(n)}{\varphi(n+1)}.
\end{equation}

\begin{lemma} 
\label{lem:site-reflection}
 We have that  
\begin{equation}\label{eq:site-reflection-identity}
 T_-R(s+h)T_0=R(s).
\end{equation}
\end{lemma}

\begin{proof}
Let $\varphi$ be an arbitrary solution  of \eqref{geig} with $\theta=s$. 
  Define
\[
 \widetilde \varphi(n)=\varphi(-n),
\]
and
\[\widetilde \Phi (n)=\vct{\widetilde \varphi(n)}{\widetilde \varphi(n-1)}\]

The  symmetry of the potential \eqref{gsym} implies that $\widetilde \varphi$ solves the same
eigen-equation, in particular by \eqref{eq:site-block-action}, one has that 

\begin{equation}\label{g1}
 \widetilde\Phi(0)=T_- \widetilde \Phi(-q).
\end{equation}

By \eqref{eq:R-action} and \eqref{eq:site-block-action},
\begin{equation}\label{g2}
    \widetilde \Phi(-q)
 =\vct{\varphi(q)}{\varphi(q+1)}
 =R(s+h)\Phi(q)
 =R(s+h)T_0\Phi(0).
\end{equation}
 
By \eqref{eq:R-action}  again,
\begin{equation}\label{g3}
    \widetilde \Phi(0)
 =\vct{\varphi(0)}{\varphi(1)}
 =R(s)\Phi(0).
\end{equation}

Multiplying \eqref{g2} by $T_-$, and by \eqref{g1} and \eqref{g3}, we have that
\[
 T_-R(s+h)T_0\Phi(0)=R(s)\Phi(0).
\]
Since the initial vector $\Phi(0)$ can be  arbitrary, \eqref{eq:site-reflection-identity}
holds.
\end{proof}

\subsection{The matrix blocks}

Put
\begin{equation}\label{eq:r-site}
 r_s=\frac{E-v(s)}{2},
 \qquad
 S_s=
 \begin{pmatrix}
 1&0\\
 r_s&1
 \end{pmatrix},
 \qquad
 D=
 \begin{pmatrix}
 1&0\\
 0&-1
 \end{pmatrix}.
\end{equation}
The motivation for introducing the matrices in \eqref{eq:r-site} may
not yet be clear at this moment. We refer readers to the   remark in Section 4.3  at the of this paper for a
conceptual explanation of their construction.

Direct computations give
\begin{equation}\label{eq:conjugate-R0}
 S_s^{-1}R(s)S_s=D,
\end{equation}
and
\begin{equation}\label{eq:conjugate-Rh}
 S_s^{-1}R(s+h)S_s
 =D_h:=
 \begin{pmatrix}
 1&0\\
 \eta_h&-1
 \end{pmatrix},
\end{equation}
where
\begin{equation*}
 \eta_h=v(s)-v(s+h)
\end{equation*}
Notice that $D_h^2=I$.  By \eqref{eq:vprime-zero} and $v\in C^2$, one has that
\begin{equation}\label{eq:eta-quadratic}
 |\eta_h|\leq Ch^2.
\end{equation}

Define
\begin{equation}\label{eq:B-site-def}
 B_0=S_s^{-1}T_0S_s,
 \qquad
 B_\pm=S_s^{-1}T_\pm S_s.
\end{equation}
By  \eqref{eq:site-reflection-identity}, \eqref{eq:conjugate-R0} and \eqref{eq:conjugate-Rh}, one has that
\begin{equation}\label{eq:Bminus-site-abstract}
 B_-=DB_0^{-1}D_h.
\end{equation}
Write
\begin{equation}\label{eq:B0-entries}
 B_0=
 \begin{pmatrix}
 a&c\\
 b&d
 \end{pmatrix}.
\end{equation}
Clearly,
\begin{equation}\label{eq:adbc}
 ad-bc=1.
\end{equation}
Equation \eqref{eq:Bminus-site-abstract} now yields the exact formula
\begin{equation}\label{eq:Bminus-site}
 B_-=
 \begin{pmatrix}
 d-\eta_hc&c\\
 b-\eta_ha&a
 \end{pmatrix}.
\end{equation}

Define 
\begin{equation}\label{eq:Gamma-site}
 \Gamma_q=B_++B_--2B_0.
\end{equation}
By Lemma~\ref{lem:second-derivative},  \eqref{eq:gamma-epsilon}, \eqref{eq:h-small}, \eqref{eq:site-blocks}, \eqref{eq:B-site-def} and the fact that
$S_s$ is fixed, we have that
\begin{equation}\label{eq:Gamma-site-small}
 \|\Gamma_q\|\leq C h^2e^{(L+\varepsilon)q} \leq C  e^{(L-2\gamma+\varepsilon)q}\to 0, \text { as } q\to\infty.
\end{equation}
In the following, all limits denoted by
$\to$ are understood as $q\to\infty$.

Similarly, by Lemma \ref{lem:uniform-upper}, \eqref{eq:gamma-epsilon}, \eqref{eq:h-small},  and \eqref{eq:eta-quadratic}, one has that 
\begin{equation}\label{eq:eta-B-small}
 |\eta_h|\,\|B_0\|
 \leq  C h^2e^{(L+\varepsilon)q} \leq C  e^{(L-2\gamma+\varepsilon)q}\to 0
\end{equation}
Assume, for contradiction, that $u\in\ell^2(\Z)$ is a nonzero eigenfunction of $H_{v,\alpha,s} u=E u$.
By Lemma~\ref{lem:simplicity} and the  symmetry of the potential \eqref{gsym}, either for all $n\in\Z$,
\[
 u(-n)=u(n)
\]
or for all $n\in\Z$,
\[
 u(-n)=-u(n).
\]

Let
\[U (n)=\vct{u(n)}{u(n-1)},\]
and 
\begin{equation*}
 X_j=S_s^{-1}U(jq),
 \qquad j=0,1,2.
\end{equation*}
Since $u\in\ell^2(\Z)$ and $q\to\infty$, we have that
$||U(q)||\to 0$ and $||U(2q)||\to 0$, so 
\begin{equation}\label{eq:X12-small}
 ||X_1||\to 0,
 \qquad
 ||X_2||\to 0.
\end{equation}

By \eqref{eq:site-block-action} and \eqref{eq:B-site-def}, one has that
\begin{equation}\label{eq:X-site-dynamics}
 X_1=B_0X_0,
 \qquad
 X_2=B_+X_1.
\end{equation}

Multiplying  \eqref{eq:Gamma-site} by $X_1$, one has that 
\begin{equation}\label{eq:Bplus-from-sym}
 (B_+-2B_0+B_-)X_1=\Gamma_q X_1.
\end{equation}

    By \eqref{eq:Gamma-site-small} and \eqref{eq:X12-small}, one has that 
\begin{equation*}
    ||\Gamma_q X_1||\to 0.
\end{equation*}
    By \eqref{eq:X12-small} and \eqref{eq:X-site-dynamics}, we have that 
\begin{equation}\label{gj183}
    ||B_+X_1||=||X_2||\to 0
\end{equation}
So, by \eqref{eq:Bplus-from-sym} and \eqref{gj183}, we conclude that 
\begin{equation} \label{gj186}
   || (2B_0-B_-)X_1||\to 0.
\end{equation}

\subsection{Contradiction for  the even eigenfunction}

Assume that $u$ is even. The equation \eqref{geig} at $n=0$ gives
\[
 2u(1)+(v(s)-E)u(0)=0,
\]
so
\[
 u(-1)=u(1)=r_su(0).
\]
Since $u$ is nontrivial, we have $u(0)\neq 0$.
Without loss of generality, assume that $u(-1)=r_s$ and $u(0)=1$.
  Therefore,
\begin{equation}\label{eq:X0-site-even}
 X_0=S_s^{-1}U(0)=\eone.
\end{equation}

By \eqref{eq:X0-site-even}, \eqref{eq:X-site-dynamics} and  \eqref{eq:B0-entries}, we have
\begin{equation}\label{eq:X1-even}
 X_1=B_0\eone=\vct{a}{b},
\end{equation}
so by \eqref{eq:X12-small},
\begin{equation}\label{gj181}
   a\to 0 \text{ and } b\to 0.
\end{equation}

Using  \eqref{eq:B0-entries}, \eqref{eq:adbc}, \eqref{eq:Bminus-site}, and
\eqref{eq:X1-even}, a direct multiplication gives
\begin{equation*}
  (2B_0-B_-)X_1=(2B_0-B_-)\vct{a}{b}
 =\vct{2a^2-1+\eta_hac}{2bd+\eta_ha^2}.
\end{equation*}

By \eqref{eq:eta-B-small} and \eqref{gj181}, one has that
\begin{equation}\label{gj184}
    |\eta_hac|
 \leq |a|\,|\eta_h|\,\|B_0\|
 \to 0.
\end{equation}

By the fact $a\to0$ and \eqref{gj184}, we have that the first coordinate of $(2B_0-B_-)X_1$,
\begin{equation*}
    2a^2-1+\eta_hac \to -1.
\end{equation*}
This  contradicts  \eqref{gj186}.

\subsection{Contradiction for the odd eigenfunction}

Assume now that $u$ is odd.  Then $u(0)=0$.  Since $u$ is nonzero,
$u(-1)\neq0$. Without loss of generality, assume $u(-1)=1$. In this case,
\begin{equation}\label{eq:X0-site-odd}
 X_0=S_s^{-1}U(0)=\etwo,
\end{equation}
and
\begin{equation*}
 X_1=B_0\etwo=\vct{c}{d},
\end{equation*}
so $c\to0$ and $d\to0$.  Again using
 \eqref{eq:B0-entries}, \eqref{eq:adbc}, and \eqref{eq:Bminus-site},
\begin{equation*}
  (2B_0-B_-)X_1=(2B_0-B_-)\vct{c}{d}
 =\vct{2ac+\eta_hc^2}{2d^2-1+\eta_hac}.
\end{equation*}

Here $d\to0$, while
\[
 |\eta_hac|
 \leq |c|\,|\eta_h|\,\|B_0\|
 \to 0.
\]
Consequently, the second coordinate of $ (2B_0-B_-)X_1$,
\begin{equation*}
 2d^2-1+\eta_hac\to -1.
\end{equation*}

This also contradicts  \eqref{gj186}.

\section{Proof of Case II}

Fix
\begin{equation*}
 \theta=x_*=s+\frac{\alpha}{2},
 \qquad
 s\in\left\{0,\frac12\right\}.
\end{equation*}
Then the  potential $V(n)=v(x_*+n\alpha)$ satisfies the following  symmetry:
\begin{align}
 v\bigl(x_*+(-n-1)\alpha\bigr)
 &=v\left(s-\left(n+\frac12\right)\alpha\right)\notag\\
 &=v\left(s+\left(n+\frac12\right)\alpha\right)
 =v(x_*+n\alpha).
\label{eq:bond-symmetry}
\end{align}
This implies for all $n$,
\begin{equation}\label{gj191}
   V(-n-1)=V(n).
\end{equation}
For a resonant scale $q$, define
\begin{equation}\label{eq:bond-blocks}
 T_0=A_q(x_*),
 \;\;
 T_+=A_q(x_*+h)=A_q(x_*+q\alpha),\;\;
 T_-=A_q(x_*-h)=A_q(x_*-q\alpha).
\end{equation}
Let
\begin{equation*}
 P=\begin{pmatrix}0&1\\1&0\end{pmatrix}.
\end{equation*}

\begin{lemma}
\label{lem:bond-reflection}
The matrices in \eqref{eq:bond-blocks} satisfy
\begin{equation}\label{eq:bond-reflection-identity}
 T_-PT_0=P.
\end{equation}
\end{lemma}

\begin{proof}
Let $\varphi$ be an arbitrary solution of \eqref{geig} with $\theta=x_*$.  Define
\[
 \widetilde \varphi(n)=\varphi(-n-1),
\]
and
\[
 \widetilde \Phi(n)
=\vct{\widetilde\varphi(n)}{\widetilde\varphi(n-1)}.
\]
By \eqref{eq:bond-symmetry} (or \eqref{gj191}), $\widetilde \varphi$ solves the same equation \eqref{geig}.  Since
$\Phi(q)=T_0\Phi(0)$,
\[
 \widetilde \Phi(-q)
 =\vct{\varphi(q-1)}{\varphi(q)}
 =P\Phi(q)
 =PT_0\Phi(0).
\]
Moreover,
\[
 \widetilde \Phi(0)
 =\vct{\varphi(-1)}{\varphi(0)}
 =P\Phi(0).
\]
Since the transfer matrix from $-q$ to $0$ begins at phase
$x_*-q\alpha=x_*-h\pmod1$ equals $T_-$, we have that 
$ \widetilde \Phi(0)=T_-\widetilde \Phi(-q) $ and hence
\[
 T_-PT_0\Phi(0)=P\Phi(0).
\]
This implies \eqref{eq:bond-reflection-identity}.
\end{proof}

\subsection{The matrix blocks}

Let
\begin{equation*}
 S=\begin{pmatrix}1&1\\1&-1\end{pmatrix},
 \qquad
 D=\begin{pmatrix}1&0\\0&-1\end{pmatrix}.
\end{equation*}
Then
\begin{equation}\label{eq:P-diagonalized}
 S^{-1}PS=D.
\end{equation}
Define
\begin{equation*}
 B_0=S^{-1}T_0S,
 \qquad
 B_\pm=S^{-1}T_\pm S.
\end{equation*}
By \eqref{eq:bond-reflection-identity} and \eqref{eq:P-diagonalized}, we have that
\begin{equation*}
 B_-=DB_0^{-1}D.
\end{equation*}
Write, as before,
\begin{equation}\label{gj192}
    B_0=\begin{pmatrix}a&c\\b&d\end{pmatrix},
 \qquad ad-bc=1. 
\end{equation}

Then
\begin{equation}\label{eq:Bminus-bond}
 B_-=\begin{pmatrix}d&c\\b&a\end{pmatrix}.
\end{equation}

\subsection{ Contradictions}

Suppose $u\in\ell^2(\Z)$ is a nonzero eigenfunction.  By
Lemma~\ref{lem:simplicity} and  reflection symmetry \eqref{gj191}, there exists
$\sigma\in\{1,-1\}$ such that
\begin{equation*}
 u(-n-1)=\sigma u(n).
\end{equation*}
In particular,
\[
 U(0)=u(0)\vct{1}{\sigma}.
\]
Clearly, $u(0)\neq 0$. Without loss of generality, assume $u(0)=1$.

Set
\begin{equation*}
 X_j=S^{-1}U(jq),
 \qquad j=0,1,2.
\end{equation*}
As in the  Case I, 
\begin{equation*}
 X_1=B_0X_0,
 \qquad
 X_2=B_+X_1,
\end{equation*}
and  
\begin{equation*}
 ||X_1||\to 0,
 \qquad
 ||X_2||\to0.
\end{equation*}
Also
\begin{equation} \label{gj186new}
   || (2B_0-B_-)X_1||\to 0.
\end{equation}
Computations show
\begin{equation}\label{eq:X0-bond}
 X_0=\eone\quad\text{if }\sigma=1,
 \qquad
 X_0=\etwo\quad\text{if }\sigma=-1.
\end{equation}

If $\sigma=1$, then $X_1=(a,b)^T$, so $a,b\to0$.  Using
\eqref{gj192} and \eqref{eq:Bminus-bond}, one obtains
\begin{equation*}
 (2B_0-B_-)X_1
 =\vct{2a^2-1}{2bd}.
\end{equation*}
Thus the first coordinate of $(2B_0-B_-)X_1 $  goes to  $-1$, contradicting \eqref{gj186new}.

If $\sigma=-1$, then $X_1=(c,d)^T$, so $c,d\to0$.  The same computation
gives
\begin{equation*}
 (2B_0-B_-)X_1
 =\vct{2ac}{2d^2-1}.
\end{equation*}
Thus the second coordinate of $(2B_0-B_-)X_1$ goes to $-1$, again contradicting
\eqref{gj186new}.

\subsection{ Remark}
In the proof, instead of estimating $T_-+T_+-2T_0$,
we conjugate the matrices $T_0$ and $T_{\pm}$ to $B_0$ and $B_{\pm}$:
$
B_0=S^{-1}T_0S,
B_{\pm}=S^{-1}T_{\pm}S,
$
and estimate $B_-+B_+-2B_0$,
The matrix $S$ is chosen so that the initial vectors associated with the $\ell^2(\Z)$ solution are mapped to the standard
basis vectors $e_1$ and $e_2$; see
\eqref{eq:X0-site-even}, \eqref{eq:X0-site-odd}, and
\eqref{eq:X0-bond}. 

We could also conduct the argument directly in terms of
$T_0$ and $T_{\pm}$. Such a formulation is somewhat more coordinate-free and may be conceptually natural. However, the
resulting matrix calculations are less transparent and more difficult for the readers
to verify. 
We therefore work with $B_0$ and $B_{\pm}$ so that  computations become explicit.

\section*{Acknowledgments}
This work was supported in part by NSF DMS-2246031.

\section*{Statements and Declarations}
{\bf Conflict of Interest} 
The author declares no conflicts of interest.

\vspace{0.2in}
{\bf Data Availability}
Data sharing is not applicable to this article as no new data were created or analyzed in this study.

\bibliographystyle{alpha} 
\bibliography{main}

@article {aj09,
    AUTHOR = {Avila, Artur and Jitomirskaya, Svetlana},
     TITLE = {The {T}en {M}artini {P}roblem},
   JOURNAL = {Ann. of Math. (2)},
  FJOURNAL = {Annals of Mathematics. Second Series},
    VOLUME = {170},
      YEAR = {2009},
    NUMBER = {1},
     PAGES = {303--342},
      ISSN = {0003-486X,1939-8980},
   MRCLASS = {47B80 (37A99 39A70 47A10 47B36 81Q10)},
  MRNUMBER = {2521117},
MRREVIEWER = {David\ Damanik},
       DOI = {10.4007/annals.2009.170.303},
       URL = {https://doi.org/10.4007/annals.2009.170.303},
}

@article {puig,
    AUTHOR = {Puig, Joaquim},
     TITLE = {Cantor spectrum for the almost {M}athieu operator},
   JOURNAL = {Comm. Math. Phys.},
  FJOURNAL = {Communications in Mathematical Physics},
    VOLUME = {244},
      YEAR = {2004},
    NUMBER = {2},
     PAGES = {297--309},
      ISSN = {0010-3616,1432-0916},
   MRCLASS = {11K60 (39A70)},
  MRNUMBER = {2031032},
MRREVIEWER = {David\ Damanik},
       DOI = {10.1007/s00220-003-0977-3},
       URL = {https://doi.org/10.1007/s00220-003-0977-3},
}

@incollection {jicm,
    AUTHOR = {Jitomirskaya, Svetlana},
     TITLE = {One-dimensional quasiperiodic operators: global theory,
              duality, and sharp analysis of small denominators},
 BOOKTITLE = {I{CM}---{I}nternational {C}ongress of {M}athematicians. {V}ol.
              2. {P}lenary lectures},
     PAGES = {1090--1120},
 PUBLISHER = {EMS Press, Berlin},
      YEAR = {[2023] \copyright 2023},
      ISBN = {978-3-98547-060-0; 978-3-98547-560-5; 978-3-98547-058-7},
   MRCLASS = {47B36 (37C55 37H15 47-02 82B26)},
  MRNUMBER = {4680277},
       DOI = {10.4171/ICM2022/175},
       URL = {https://doi.org/10.4171/ICM2022/175},
}

@article{gor,
  title={The point spectrum of the one-dimensional {S}chr\"odinger operator},
  author={Gordon, A Ya},
  journal={Uspekhi Matematicheskikh Nauk},
  volume={31},
  number={4},
  pages={257--258},
  year={1976},
  publisher={Russian Academy of Sciences, Steklov Mathematical Institute of Russian~…}
}

@article{sim1,
  title={Almost periodic {S}chr{\"o}dinger operators: a review},
  author={Simon, Barry},
  journal={Advances in Applied Mathematics},
  volume={3},
  number={4},
  pages={463--490},
  year={1982},
  publisher={Elsevier}
}

@article {as82,
    AUTHOR = {Avron, Joseph and Simon, Barry},
     TITLE = {Singular continuous spectrum for a class of almost periodic
              {J}acobi matrices},
   JOURNAL = {Bull. Amer. Math. Soc. (N.S.)},
  FJOURNAL = {American Mathematical Society. Bulletin. New Series},
    VOLUME = {6},
      YEAR = {1982},
    NUMBER = {1},
     PAGES = {81--85},
      ISSN = {0273-0979},
   MRCLASS = {47B37 (81Cxx)},
  MRNUMBER = {634437},
       DOI = {10.1090/S0273-0979-1982-14971-0},
       URL = {https://doi.org/10.1090/S0273-0979-1982-14971-0},
}

@article {js,
    AUTHOR = {Jitomirskaya, S. and Simon, B.},
     TITLE = {Operators with singular continuous spectrum. {III}. {A}lmost
              periodic {S}chr\"{o}dinger operators},
   JOURNAL = {Comm. Math. Phys.},
  FJOURNAL = {Communications in Mathematical Physics},
    VOLUME = {165},
      YEAR = {1994},
    NUMBER = {1},
     PAGES = {201--205},
      ISSN = {0010-3616},
   MRCLASS = {47-02 (34L05 47E05 81Q10)},
  MRNUMBER = {1298948},
MRREVIEWER = {Horst Behncke},
       URL = {http://projecteuclid.org/euclid.cmp/1104271040},
}

@article {jl,
    AUTHOR = {Jitomirskaya, Svetlana and Liu, Wencai},
     TITLE = {Arithmetic spectral transitions for the {M}aryland model},
   JOURNAL = {Comm. Pure Appl. Math.},
  FJOURNAL = {Communications on Pure and Applied Mathematics},
    VOLUME = {70},
      YEAR = {2017},
    NUMBER = {6},
     PAGES = {1025--1051},
      ISSN = {0010-3640},
   MRCLASS = {47A10 (37A45 39A60 82B10)},
  MRNUMBER = {3639318},
MRREVIEWER = {Sophie Grivaux},
       DOI = {10.1002/cpa.21688},
       URL = {https://doi.org/10.1002/cpa.21688},
}

@article {jl1,
    AUTHOR = {Jitomirskaya, Svetlana and Liu, Wencai},
     TITLE = {Universal hierarchical structure of quasiperiodic
              eigenfunctions},
   JOURNAL = {Ann. of Math. (2)},
  FJOURNAL = {Annals of Mathematics. Second Series},
    VOLUME = {187},
      YEAR = {2018},
    NUMBER = {3},
     PAGES = {721--776},
      ISSN = {0003-486X},
   MRCLASS = {47B36 (37C55 82D40)},
  MRNUMBER = {3779957},
MRREVIEWER = {T\'{u}lio O. Carvalho},
       DOI = {10.4007/annals.2018.187.3.3},
       URL = {https://doi.org/10.4007/annals.2018.187.3.3},
}

@article{jl2,
   AUTHOR = {Jitomirskaya, Svetlana and Liu, Wencai},
     TITLE = {Universal reflective-hierarchical structure of quasiperiodic
              eigenfunctions and sharp spectral transition in phase},
   JOURNAL = {J. Eur. Math. Soc. (JEMS)},
  FJOURNAL = {Journal of the European Mathematical Society (JEMS)},
    VOLUME = {26},
      YEAR = {2024},
    NUMBER = {8},
     PAGES = {2797--2836},
      ISSN = {1435-9855},
   MRCLASS = {47B36 (37C55 81Q10)},
  MRNUMBER = {4756946},
       DOI = {10.4171/jems/1325},
       URL = {https://doi.org/10.4171/jems/1325},
}

@article {ayz,
    AUTHOR = {Avila, Artur and You, Jiangong and Zhou, Qi},
     TITLE = {Sharp phase transitions for the almost {M}athieu operator},
   JOURNAL = {Duke Math. J.},
  FJOURNAL = {Duke Mathematical Journal},
    VOLUME = {166},
      YEAR = {2017},
    NUMBER = {14},
     PAGES = {2697--2718},
      ISSN = {0012-7094},
   MRCLASS = {47B36 (37C55 39A70 47B39 81Q10 82B26)},
  MRNUMBER = {3707287},
MRREVIEWER = {Peicheng Zhu},
       DOI = {10.1215/00127094-2017-0013},
       URL = {https://doi.org/10.1215/00127094-2017-0013},
}

@article{liupafa,
  title={A New Proof of the Sharp {G}ordon's Lemma: No Eigenvalues for {S}chr{\" o}dinger Operators with Almost Repetition Potentials},
  author={Liu, Wencai},
  journal={Pure Appl. Funct. Anal. to appear}}

@article {liupmj,
    AUTHOR = {Liu, Wencai},
     TITLE = {Small denominators and large numerators of quasiperiodic
              {S}chr\"odinger operators},
   JOURNAL = {Peking Math. J.},
  FJOURNAL = {Peking Mathematical Journal},
    VOLUME = {8},
      YEAR = {2025},
    NUMBER = {3},
     PAGES = {503--532},
      ISSN = {2096-6075,2524-7182},
   MRCLASS = {47A10 (47B39 81Q10)},
  MRNUMBER = {4950602},
MRREVIEWER = {Atsuhide\ Ishida},
       DOI = {10.1007/s42543-023-00075-3},
       URL = {https://doi.org/10.1007/s42543-023-00075-3},
}

\end{document}